\documentclass[a4 paper, 12pt]{article}

\usepackage{arxiv}

\usepackage[utf8]{inputenc} 
\usepackage[T1]{fontenc}    
\usepackage{hyperref}       
\usepackage{url}            
\usepackage{booktabs}       
\usepackage{amsmath, amsfonts, amssymb, amsthm, mathtools}       
\usepackage{nicefrac}       
\usepackage{microtype}      
\usepackage{lipsum}
\usepackage{graphicx}
\usepackage{wrapfig}
\usepackage{subcaption}
\usepackage{tikz}
\usepackage[inline]{enumitem}
\usepackage{float}

\usepackage{mathtools}          
\usepackage{enumitem}           
\usepackage{listings}           
\usepackage{todonotes}         

\usepackage{doi}
\usepackage{algorithm, algpseudocode}
\usepackage{algcompatible}
\usepackage{color}
\usepackage{wrapfig}
\usepackage{tcolorbox}
\usepackage{xspace}
\usepackage{breqn}

\usepackage[backend=bibtex,style=alphabetic,natbib=true, maxbibnames=99]{biblatex}
\usepackage[autostyle=true]{csquotes}

\newtheorem{theorem}{Theorem}[section]
\newtheorem{corollary}{Corollary}[theorem]
\newtheorem{lemma}[theorem]{Lemma}

\theoremstyle{definition}
\newtheorem{definition}[theorem]{Definition}

\theoremstyle{remark}
\newtheorem{remark}[theorem]{Remark}
\newtheorem{example}[theorem]{Example}

\allowdisplaybreaks[1]

\def\real{{\mathbb R}}
\def\nonnegreal{{\mathbb R_{\ge 0}}}

\def\int{{\mathbb Z}}
\def\positiveint{{\mathbb Z_{> 0}}}
\def\game{{\mathcal{G}}}
\DeclareMathOperator{\carpro}{\square}

\def\cc#1{\mathtt{#1}}

\DeclareMathOperator{\NP}{\cc{NP}}

\DeclareMathOperator{\PPAD}{\cc{PPAD}}

\DeclareMathOperator{\PLS}{\cc{PLS}}
\DeclareMathOperator{\PPADPLS}{\cc{PPAD}\cap\cc{PLS}}

\title{On Finding Pure Nash Equilibria of Discrete Preference Games and Network Coordination Games}
\shorttitle{On Finding PNEs of Discrete Preference Games \& Network Coordination Games}

\author{
 Takashi Ishizuka\\
  Graduate School of Mathematics,\\
  Kyushu University,\\
  744 Motooka, Nishi-ku, Fukuoka, Japan\\
  \texttt{ishizuka.takashi.664@s.kyushu-u.ac.jp} \\
   \And
 Naoyuki Kamiyama\\
  Institute of Mathematics for Industry,\\
  Kyushu University,\\
  744 Motooka, Nishi-ku, Fukuoka, Japan\\
  \texttt{kamiyama@imi.kyushu-u.ac.jp} \\
}

\begin{document}
\maketitle
\begin{abstract}
	This paper deals with the complexity of the problem of computing a pure Nash equilibrium for discrete preference games and network coordination games beyond $O(\log n)$-treewidth and tree metric spaces.
	First, we estimate the number of iterations of the best response dynamics for a discrete preference game on a discrete metric space with at least three strategies.
	Second, we present a sufficient condition that we have a polynomial-time algorithm to find a pure Nash equilibrium for a discrete preference game on a grid graph.
	Finally, we discuss the complexity of finding a pure Nash equilibrium for a two-strategic network coordination game whose cost functions satisfy submodularity. In this case, if every cost function is symmetric, the games are polynomial-time reducible to a discrete preference game on a path metric space.
\end{abstract}


\section{Introduction} \label{SecIntroduction}
A {\it graphical game}, introduced by Kearns et al.\ \cite{KLS01}, is a succinctly represented multi-player strategic form game. 
A  graphical game consists of an undirected graph $G = (V, E)$, where $V$ is a finite set of players and every edge in $E$ represents the interaction between its endpoints, and for each player $i \in V$, a finite set of strategies $S_i$ and a cost function $C_{i} \colon \prod_{j \in N(i) \cup \{ i \} } S_{j} \to \nonnegreal$, where $N(i)$ is the set of neighbors of player $i$, that is, $N(i) = \{ j \in V ; \{i, j\} \in E\}$.

We refer to the tuple of strategies played by each player as a strategy profile. The set $S = \prod_{i \in [n]}S_{i}$ is called a set of strategy profiles.
For a player $i \in V$, we denote by $S_{-i}$ the set of strategies for all players except $i$.
For a strategy profile $x = (x_i)_{i \in V}$, we denote by $x_i$ the strategy played by a player $i$, and by $x_{-i}$ the strategies of all players except $i$.

A pure Nash equilibrium is an intuitive and essential concept of rationality. 
A pure Nash equilibrium is a strategy profile such that every player has no incentive to change her selected strategy.
Every player aims at minimizing her cost. A strategy profile $x^* = (x_{i}^{*})_{i \in V}$ is a pure Nash equilibrium if for each player $i \in V$, and every strategy $x_i \in S_i$, $C_i(x_i^*, x^*_{-i}) \le C_i(x_i, x_{-i}^*)$ holds.
Note that graphical games do not always have pure Nash equilibria because any two-player strategic form game is a graphical game.

The complexity of finding a pure Nash equilibrium on a graphical game is one of the most interesting topics of Algorithmic Game Theory. 
Unfortunately, it is intractable to determine the existence of a pure Nash equilibrium for a graphical game. Gottlob et al.\ \cite{GGS05} have proven that the problem of deciding whether there exists a pure Nash equilibrium for a given graphical game is $\NP$-hard.
On the other hand, Daskalakis and Papadimitriou \cite{DP06} have shown that it is polynomial-time decidable whether there is a pure Nash equilibrium on a graphical game whose players' network has $O(\log n)$-treewidth. Furthermore, their result has stated that we can find a pure Nash equilibrium in polynomial time if it exists for such a graphical game.

We wish to understand what properties make it hard to compute a pure Nash equilibrium for a graphical game and make it easy to do. This paper focuses on the class of graphical games that are guaranteed the existence of pure Nash equilibria. There are well-known classes of graphical games that always have a pure Nash equilibrium; a {\it discrete preference game} and a {\it network coordination game} are examples.

A discrete preference game with a parameter $\game = (G, \mathcal{M}, ( \beta_{i} ), \alpha)$, which is the fundamental model introduced by Chierichetti et al.\ \cite{CKO18}, consists of an unweighted graph $G = (V, E)$, a finite metric space $\mathcal{M} = (L, d)$, a preferred strategy $\beta_i \in L$ for each player $i \in V$, and a parameter $0 \le \alpha < 1$. Every player has the identical strategy set $L$.
Given a strategy profile $x = (x_i)_{i \in V}$, the cost for player $i$ is:
\begin{align}
	c_i(x) = \alpha d(x_i, 	\beta_{i}) + (1 - \alpha) \sum_{ j \in N(i)} d(x_i, x_j).
\end{align}

A network coordination game $\game = (G, ( S_{i} ), ( C_{i, j}, C_{j, i} ) )$ is defined by: (i) an undirected graph $G = (V, E)$; (ii) for each edge $\{i, j\} \in E$, there are two cost functions $C_{i, j} \colon S_i \times S_j \to \nonnegreal$ and  $C_{j, i} \colon S_j \times S_i \to \nonnegreal$ that satisfy $C_{i, j}(x_i, x_j) = C_{j, i}(x_j, x_i)$ for all $x_i \in S_i$ and $x_j \in S_j$; (iii) the total cost for a player $i \in V$ is the sum of all her costs, i.e., $C_i(x) = \sum_{ j \in N(i) }C_{i, j}(x_i, x_j)$.

The results of the hardness of computing a pure Nash equilibrium for a discrete preference game and a network coordination game are known. Lolakapuri et al.\ \cite{LBNPD19} have proven that finding a pure Nash equilibrium on a discrete preference game is $\PLS$-complete even if the maximum degree of the players' network is $7$. Cai and Daskalakis \cite{DP06} have shown the $\PLS$-completeness of computing a pure Nash equilibrium for a network coordination game even if the maximum degree of the players' network is five and each player has two strategies.
On the other hand, Lolakapuri et al.\ \cite{LBNPD19} have proven that a pure Nash equilibrium for a discrete preference game on a tree metric space is polynomial-time computable.

The following computational aspects of pure Nash equilibria are still unknown for discrete preference games and network coordination games:
\begin{itemize}
	\item How hard is computing a pure Nash equilibrium for a discrete preference game on a non-tree metric space?
	\item Can we find pure Nash equilibria in polynomial time for a discrete preference game and a network coordination game if the maximum degree of the players' networks is four?
\end{itemize}
This paper deals with the above topics. In particular, we discuss the complexity of finding a pure Nash equilibrium for a discrete preference game on neither $O(\log n)$-treewidth nor a tree metric space.
First, we estimate an upper bound of the number of iterations of the best response dynamics for a discrete preference game on a discrete metric space to compute a pure Nash equilibrium.
Second, we provide a sufficient condition that we have a polynomial-time algorithm to find a pure Nash equilibrium of such a discrete preference game.
Finally, we present a relationship between discrete preference games and network coordination games.

\subsection{Our Results}
\paragraph{Discrete preference game on the discrete metric}
A discrete metric space with at least three strategies is one of the simple non-tree metric spaces.
It is important to consider and understand the complexity of a discrete preference game with such a metric space. Recall that a discrete preference game was formulated based on a decision-making model wherein agents decide which platform to use \cite{LBNPD19}. Note that the metric space implies that every agent is only interested in being on the same or different platforms. Namely, a discrete preference game on a discrete metric space is one of the uncomplicated settings of decision-making models. 

Section \ref{SecDiscreteMetric} provides an upper bound for the number of iterations of the best response dynamics for a discrete preference game on the discrete metric.
We show that the best response dynamics halts after quadratic iterations when we view the given parameter as a constant.

\paragraph{Discrete preference games on grid graph}
Our motive behind this work is to clarify the boundary between cases where we can find a pure Nash equilibrium in polynomial time for the numbers of players and strategies and cases where it is not\footnote{Note that the games dealt with in this paper are guaranteed the existence of pure Nash equilibria. This fact implies that we can trivially find it in polynomial time when we regard the number of players as a constant. On the other hand, it is not always possible to compute a pure Nash equilibrium in polynomial time when the number of strategies is considered a constant.}.
As mentioned above, the complexity of finding a pure Nash equilibrium on a graph with degree four is unknown for discrete preference games and network coordination games. Hence, it is important to clarify the complexity of finding a pure Nash equilibrium for a discrete preference game on a two-dimensional grid graph. A two-dimensional grid graph is one of the graphs whose maximum degree is four.

Remark that Section \ref{SecReductionfromDPGtoNCG} shows the relationship between discrete preference games and network coordination games. In particular, we prove that there is a polynomial-time reduction such that the structure of the players' network is preserved from a discrete preference game to a network coordination game.
This fact implies that the hardness result for a network coordination game straightforwardly follows from the hardness results for a discrete preference game. Therefore, it is a natural approach to deal with the complexity of discrete preference games first, under negative conjecture.

Section \ref{SecDPGonGridGraph} provides a sufficient condition that we have a polynomial-time algorithm to find a pure Nash equilibrium of a discrete preference game on a grid graph.
To prove this condition, Section \ref{SecCartesianGame} introduces a more general discrete preference game, called a cartesian game, in which a discrete preference game is constructed from some discrete preference games.
We show that it can efficiently construct a pure Nash equilibrium for a cartesian game from pure Nash equilibria for the discrete preference games that form that cartesian game.
Our results are the first polynomial-time computability of discrete preference games on neither $O(\log n)$-treewidth nor tree metric spaces.

\subsection{Related Works}
Els\"{a}sser and Tscheuschner \cite{ET11} have proved that the problem of computing a pure Nash equilibrium for a \textsc{Mac-Cut} game, which is a special case of network coordination games, is $\PLS$-complete even if the maximum degree of the players' network is five. Remark that the $\PLS$-hardness of a discrete preference game and a network coordination game relies on the $\PLS$-hardness of a \textsc{Mac-Cut} game \cite{CD11, LBNPD19}.
There is a positive result about computing a pure Nash equilibrium on a network coordination game: Poljak \cite{Pol95} has proven that we have a polynomial-time algorithm for computing a pure Nash equilibrium of a \textsc{Mac-Cut} game whose players' network is a cubic graph.

It is well-known that every graphical game always has a mixed Nash equilibrium. However, it is also hard to find a mixed Nash equilibrium.
Chen et al.\ \cite{CDT09} proved that the problem of finding a mixed Nash equilibrium on a polymatrix game is $\PPAD$-complete. Even if the players' network is a tree, computing a mixed Nash equilibrium is still hard \cite{DFS20}.
Naturally, not all problems of finding a mixed Nash equilibrium are intractable. 
Cai et al.\ \cite{CCDP16} have proven that a pure Nash equilibrium on a zero-sum polymatrix game can be found in polynomial time.
Elkind et al.\ \cite{EGG06} have shown that we can compute a mixed Nash equilibrium for a polymatrix game when a players' network is a path and each player has two strategies.

It seems easy to compute a mixed Nash equilibrium on a network coordination game. Cai and Daskalakis \cite{CD11} have pointed out that such a problem belongs to $\PPADPLS$. However, it is still unknown which is true: We have a polynomial-time algorithm for finding a mixed Nash equilibrium of a network coordination game, or such a problem is $\PPADPLS$-complete.
Babichenko and Rubinstein \cite{BR21} have proven the $\PPADPLS$-completeness of a {\it polytensor identical interest game}, which is a generalization of a network coordination game.

\section{Preliminaries}
\paragraph{Basic Notations}
We denote by $\positiveint$ and $\nonnegreal$ the sets of positive integers and non-negative real numbers, respectively.
We use $[n] = \{1, 2, \dots, n \}$ for $n \in \positiveint$.
A space $(L, d)$ is a metric space if the function $d \colon L \times L \to \nonnegreal$ satisfies the following conditions:  for all $x, y, z \in L$, (i) $d(x, y) = 0$ if and only if $x = y$; (ii) $d(x, y) = d(y, x)$; and (iii) $d(x, y) \le d(x, z) + d(z, y)$.
Specifically, we refer to a metric space $(L, d)$ whose distance satisfies that $d(x, y) = 1$ whenever $x \neq y$ as a discrete metric.

A graph metric is represented by an edge-weighted undirected graph. The distance between any pair of points is the weight of the minimum weight path in the graph between the corresponding vertices. A graph metric is a tree or path metric if the graph is a tree or a path, respectively.

For some $\ell \in \positiveint \cup \{ \infty \}$, a strategy space $\mathcal{M} = (L, d)$ is the $\ell$-product metric space of $k$ metric spaces $\mathcal{M}_1 = (L_1, d_1), \dots, \mathcal{M}_k = (L_k, d_k)$ if $L = L_1 \times \cdots \times L_k$ and for any two points $x = (x^1, \dots, x^k)$ and $y = (y^1, \dots, y^k)$ in $L$, the distance $d(x, y)$ is defined as $\| ( d_1(x^1, y^1), \dots, d_k(x^k, y^k) ) \|_{\ell}^{\ell}$, where $\| \cdot \|_{\ell}$ means the $\ell$-norm if $\ell \in \positiveint$; otherwise we define $d(x, y) = \max_{ t \in [k] } \{ d_{t}(x^t, y^t) \}$.

\paragraph{Potential Games}
A game is an {\it exact potential game} if there exists a function $\Phi \colon S \to \real$ such that for all $s_{-i} \in S_{ - i }$, $s_{i}, t_{i} \in S_i$, $\Phi(s_i, s_{-i}) - \Phi(t_i, s_{-i}) = C_{i}(s_i, s_{ - i}) - C_{i}(t_{i}, s_{-i})$, where $S$ is the set of strategy profiles, and  $C_{i}$ is the cost for player $i$. In this paper, we call such a function an exact potential function for the game.
A game is a {\it generalized ordinal potential game} if there is a function $\Phi \colon S \to \real$ such that for all $s_{-i} \in S_{ - i }$, $s_{i}, t_{i} \in S_i$, $\Phi(s_i, s_{-i}) > \Phi(t_i, s_{-i})$ whenever $C_{i}(s_i, s_{ - i}) > C_{i}(t_{i}, s_{-i})$. We call such a function a generalized ordinal potential function for the game.
The existence of pure Nash equilibrium for some variants of potential games can be found in Chapter 2.2 of \cite{LHS16}.
Notice that we can easily see that an exact potential game always has a pure Nash equilibrium since the best response dynamics, which is described in Section \ref{SecDiscreteMetric}, halts after the finite number of iterations.

\section{Discrete Preference Games on Discrete Metric Spaces} \label{SecDiscreteMetric}
In this section, we estimate an upper bound of the number of iterations of the best response dynamics for a discrete preference game with a parameter on a discrete metric space.
Recall that the two-strategic case was studied by previous work \cite{CKO18, FGV16}. 
We now focus on the case when there are three or more strategies.

Let $\game = \big( G = (V, E), \mathcal{M} = (L, d), (\beta_{i})_{ i \in V}, \alpha \big)$ be a discrete preference game with a parameter $0 \le \alpha < 1$, where the metric space $\mathcal{M}$ is a discrete metric space.
We define the potential $\Phi$ for a strategy profile $x = (x_i)_{i \in V}$ as 
\begin{align} \label{Eq:ExactPotentialFunction-DPGwithParameter}
	\Phi(x) = \sum_{i \in V} \alpha d(x_i, \beta_i) + ( 1 - \alpha ) \sum_{ \{ i, j \} \in E} d(x_i, x_j).
\end{align}
Note that the above function $\Phi$ is an exact potential function for $\game$ \cite{CKO18}.
Therefore, any player decreases her cost if and only if the potential $\Phi$ also decreases by the same value.

The best response dynamics follows the following procedure:
	While the current strategy profile is not a pure Nash equilibrium, pick an arbitrary player who wants to deviate from the current strategy profile, and she will change her strategy to a best response.
Note that there is only one player moving strategy at each step in the best response dynamics.

The following theorem gives an upper bound of the number of iterations of the best response dynamics for a discrete preference game with a parameter.
Then we denote by $\Phi_{\max}$ the potential of the start point.  

\begin{theorem} \label{TheoremUpperBoundforBestResponseDynamics}
The best response dynamics for a discrete preference game on a discrete metric space halts after at most $\mu(\alpha)^{-1} \Phi_{\max}$ steps, where $\mu(\alpha) = \min\{ 1 - \alpha,  \alpha + (1 - \alpha) \lfloor 1 - \alpha/(1 - \alpha) \rfloor, -\alpha + (1 - \alpha) \lfloor 1 + \alpha/(1 - \alpha) \rfloor \}$, and $\Phi_{\max}$ is the potential of the start point.
\end{theorem}
\begin{proof}
	We denote by $D_i(x)$ the number of neighbors that plays a strategy different from $i$'s strategy, i.e., $D_i(x) = | \{ j \in N(i) ; x_i \neq x_j\}|$.
	Then given a strategy profile $x = (x_v)_{ v \in V}$, the cost of player $i$ is: $c_i(x) = \alpha d(x_i, \beta_i) + (1 - \alpha) D_i(x)$.
	
	We consider the best response dynamics.
	Let $x = (x_i)_{v \in V}$ be a current strategy. In this step, a player $i \in V$ moves her strategy from $x_i$ to $y_i$, and $i$'s cost strictly decreases.
	Then, there are three possible cases:
	\begin{itemize}
		\item If $x_i \neq \beta_i \neq y_i$, then we have $0 < c_i(x_i, x_{-i}) - c_i(y_i, x_{ - i}) = (1-\alpha) (D_i(x_i, x_{-i}) -D_i(y_i, x_{-i}))$. In this case, it satisfies that $D_i(x_i, x_{-i}) -D_i(y_i, x_{-i}) > 0$. Notice that $D_i( \cdot )$ is a non-negative integer, and hence, $\Phi(x_i, x_{-i}) - \Phi(y_i, x_{ - i}) = c_i(x_i, x_{-i}) - c_i(y_i, x_{ - i}) \ge (1 - \alpha) > 0$ holds.
		\item If $x_i \neq \beta_i = y_i$, then we have $0 < c_i(x_i, x_{-i}) - c_i(y_i, x_{ - i}) = \alpha + (1-\alpha) (D_i(x_i, x_{-i}) -D_i(y_i, x_{-i}))$. In this case, it satisfies that  $D_i(x_i, x_{-i})  - D_i(y_i, x_{-i}) > - \alpha/(1 - \alpha)$. Note that $D_{i}(\cdot)$ is a non-negative integer. If $- \alpha/(1 - \alpha)$ is an integer, then it holds that $D_i(x_i, x_{-i})  - D_i(y_i, x_{-i}) \ge 1 - \alpha/(1 - \alpha)$, and otherwise, it holds that $- \alpha/(1 - \alpha) < \lceil - \alpha/(1 - \alpha) \rceil = \lfloor 1 - \alpha/(1 - \alpha) \rfloor \le D_i(x_i, x_{-i})  - D_i(y_i, x_{-i})$. Hence, we have $D_i(x_i, x_{-i})  - D_i(y_i, x_{-i}) \ge \lfloor 1 - \alpha/(1 - \alpha) \rfloor$. This implies that $\Phi(x_i, x_{-i}) - \Phi(y_i, x_{ - i}) = c_i(x_i, x_{-i}) - c_i(y_i, x_{ - i}) \ge \alpha + (1 - \alpha) \lfloor 1 - \alpha/(1 - \alpha) \rfloor > 0$.
		\item If $x_i = \beta \neq y_i$, then we have  $0 < c_i(x_i, x_{-i}) - c_i(y_i, x_{ - i}) = -\alpha + (1-\alpha) (D_i(x_i, x_{-i}) -D_i(y_i, x_{-i}))$. In this case, it satisfies that $D_i(x_i, x_{-i}) - D_i(y_i, x_{-i}) > \alpha/(1 - \alpha)$. Note that $D_{i}(\cdot)$ is a non-negative integner. If $\alpha/(1 - \alpha)$ is an integer, then it holds that $D_i(x_i, x_{-i}) - D_i(y_i, x_{-i}) \ge 1 + \alpha/(1 - \alpha)$, and otherwise, it holds that $\alpha/(1 - \alpha) < \lceil \alpha/(1 - \alpha) \rceil \le \lfloor 1 + \alpha/(1 - \alpha) \rfloor \le D_i(x_i, x_{-i}) - D_i(y_i, x_{-i})$. Therefore, we have $D_i(x_i, x_{-i}) - D_i(y_i, x_{-i}) \ge \lfloor 1 + \alpha/(1 - \alpha) \rfloor$. This implies that $\Phi(x_i, x_{-i}) - \Phi(y_i, x_{ - i}) = c_i(x_i, x_{-i}) - c_i(y_i, x_{ - i}) \ge -\alpha + (1 - \alpha) \lfloor 1 + \alpha/(1 - \alpha) \rfloor > 0$.
	\end{itemize}
	From the above observation, at each step of the best response dynamics, the potential $\Phi$ decreases at least $\mu(\alpha)$.
	Therefore, the best response dynamics halts after at most $\mu(\alpha)^{-1} \Phi_{\max}$ steps, where $\Phi_{\max}$ is the potential for the start point.
\end{proof}

\begin{remark}
When we view the given parameter $\alpha$ as a constant, the best response dynamics halts after at most $O(n^2)$ iterations by Theorem \ref{TheoremUpperBoundforBestResponseDynamics} because the exact potential function $\Phi(x)$ in $O(n^2)$. Note that it halts after at most $O(n)$ iterations in the two-strategic setting by a technical way to select a player who moves her strategy at each step even if the parameter $\alpha$ is non-constant \cite{CKO18}.
\end{remark}

\section{Discrete Preference Games on Grid Graphs} \label{SecDPGonGridGraph}
We present the special case of a discrete preference game whose pure Nash equilibria can be found in polynomial time beyond $O(\log n)$-treewidth and tree metrics. 

We consider a discrete preference game with a parameter on a $k$-dimensional grid graph.
We call a graph $G = (V, E)$ a $k$-dimensional grid graph if there are $k$ positive integers $M_1, \dots, M_k$ such that $V = [M_1] \times \cdots \times [M_k]$ and there is an edge $\{ i, j \} \in E$ if $\| i - j \|_{1} = 1$.

Now, we prove that there is a polynomial-time algorithm to compute a pure Nash equilibrium for a discrete preference game $\game = (G, \mathcal{M}, (\beta_{i})_{i \in V}, \alpha)$ on $k$-dimensional grid graph $G$ if the following two conditions hold:
\begin{enumerate}
\renewcommand{\theenumi}{(\Alph{enumi})}
\renewcommand{\labelenumi}{\theenumi}
	\item  \label{ConditionA} $\mathcal{M} = (L, d)$ is a $1$-product metric space of $k$ arbitrary finite metric spaces $\mathcal{M}_1 = (L_1, d_1),\linebreak \dots, \mathcal{M}_k = (L_k, d_k)$; and
	\item  \label{ConditionB} we can select a strategy $\beta_{i_t}^{t} \in L_{t}$ for each $t \in [k]$ and each $i_{t} \in [M_t]$ so that the set $\{ \beta_{i_t}^{t} \in L_{t} ~;~ t \in [k], i_{t} \in [M_t] \}$ satisfies the following condition: for each player $i = (i_1, \dots, i_k) \in V$, the preferred strategy $\beta_{i}$ is a form of  $(\beta_{i_{1}}^{1}, \dots, \beta_{i_{k}}^{k}) \in L$.
\end{enumerate}
In other words, the second condition implies that the $t$-th element of the preferred strategy $\beta_{i}$ of a player $i$ is $\beta_{i_{t}}^{t}$ whenever the $t$-th player of $i$ is $i_{t}$.
For instance, we consider a discrete preference game on a two-dimentional grid graph $G = ([N_1] \times [N_2], E)$ satisfying the above two conditions.
The condition \ref{ConditionB} implies that the player $(i_1, i_2)$ prefers the strategy $(\beta_{i_1}^{1}, \beta_{i_2}^{2})$ if a player $(i_1, j_2)$ and a player $(j_1, i_2)$ prefer strategies $(\beta_{i_1}^{1}, \beta_{j_2}^{2})$ and $(\beta_{j_1}^{1}, \beta_{i_2}^{2})$, respectively.

\begin{theorem}	\label{TheoremPolynomial-timeComputabilityGridGraph}
	We suppose that a discrete preference game $\game = (G, \mathcal{M}, (\beta_{i})_{i \in V}, \alpha)$ on $k$-dimensional grid graph $G$ satisfies the above two conditions \ref{ConditionA} and \ref{ConditionB}. In this case, we can find a pure Nash equilibrium for $\game$ in polynomial time.
\end{theorem}

To prove the above theorem, we introduce a cartesian product of discrete preference games, a game formed by some discrete preference games, in Section \ref{SecCartesianGame}.
We prove that a pure Nash equilibrium for a cartesian product of discrete preference games is efficiently constructible from pure Nash equilibria for ingredients of the original one.
After that, we give the proof of Theorem \ref{TheoremPolynomial-timeComputabilityGridGraph} in Section \ref{SecPolynomial-timeSolvabilityDPG}.

\subsection{Properties of Cartesian Products of Discrete Preference Games} \label{SecCartesianGame}
A cartesian product of discrete preference games is formed by $k$ discrete preference games with a parameter.
This model represents an environment where every player belongs to $k$ different communities and makes decisions within each community.
Here, we suppose that each community forms its own network. When we assume that every community forms the same network, such a game is a discrete preference game on a product metric space --- we discuss the complexity of such a model in Section \ref{SecGameProductMetricSpace}.

Informally speaking, a players' network on a cartesian product of discrete preference games is a cartesian product of graphs that are networks for the ingredients of the original one.
Each player is a tuple of players on ingredients, and they communicate along only one edge on an ingredient. Furthermore, a strategy space comprises a product metric space\footnote{This paper deals with a case of a $1$-product metric space, but it can also be generalized to any $\ell$-product metric space.}.

We define a cartesian product of graphs and a cartesian product of discrete preference games.

\begin{definition}[Cartesian Product of Graphs]
	Let $G_{1} = ( V_{1}, E_{1} ),$ $\dots, G_{k} = ( V_{k}, E_{k} )$ be simple graphs.
	We define the cartesian product of graphs $G = (V, E)$ as follows:
	Each node $v \in V$ is a $k$-tuple of nodes $(v_1, \dots, v_k)$, where $v_i \in V_i$ for each $i \in [k]$. There is an edge $\{ v , u \} \in E$ if and only if there exists only one $t \in [k]$ such that $\{ v_t, u_t \} \in E_{t}$ and $v_i = u_i$ for all $i \neq t$.
	We denote $G_1 \carpro G_2 \carpro \cdots \carpro G_k$ by the cartesian product of $k$ graphs $G_1, \dots, G_k$.
\end{definition}

\begin{definition} [Cartesian Product of Discrete Preference Game] \label{Def:CartesianProductofDPG}
	Fix a parameter $0 \le \alpha < 1$.
	Given $k$ discrete preference games with a parameter $\game_1 = (G_{1} = (V_1, E_1), \mathcal{M}_1 = (L_1, d_1),  (\beta_{i}^{1})_{i \in V_1}, \alpha),$ $\dots, \game_k = (G_{k} = (V_k, E_k), \mathcal{M}_k = (L_k, d_k), (\beta_{i}^{k})_{i \in V_{k}}, \alpha)$, we define the discrete preference game $\game = (G = (V, E), \mathcal{M} = (L, d), (\beta_{i})_{i \in V}, \alpha)$ as follows: the graph $G := G_1 \carpro G_2 \carpro \cdots \carpro G_k$, the strategy space $\mathcal{M}$ is a $1$-product metric space of $\mathcal{M}_{1}, \dots, \mathcal{M}_{k}$, and, for each player $i = (i_{1}, \dots, i_{k}) \in V$, the strategy profile $\beta_{i} = (\beta_{i_{1}}^{1}, \dots, \beta_{i_{k}}^{k})$. In this case, we call $\game$ the cartesian game constructed from discrete preference games $\game_{1}, \dots, \game_{k}$.
\end{definition}

Let $\game$ be the cartesian game constructed from discrete preference games $\game_{1}, \dots, \game_{k}$.
When we are given a strategy profile $x^{t}$ of $\game_{t}$ for each $t \in [k]$, we interpret $x = (x^{t})_{t \in [k]}$ as the strategy profile of $\game$ such that each player $i = (i_{1}, \dots, i_{k}) \in V$ plays the strategy $x_{i} = (x_{i_{1}}^{1}, \dots, x_{i_{k}}^{k})$.

The next theorem states that we can efficiently construct a pure Nash equilibrium for $\game$ from pure Nash equilibria for $\game_{1}, \dots, \game_{k}$.

\begin{theorem} \label{TheoremCartesianProductDPG}
	Suppose that $\game = (G, \mathcal{M}, (\beta_{i})_{i \in V}, \alpha)$ is a Cartesian game constructed from $k$ discrete preference games $\game_t = (G_{t}, \mathcal{M}_t,  (\beta_{i}^{t})_{i \in V_t}, \alpha)$ for $t \in [k]$.
	In this case, a strategy profile $\hat{x} = (\hat{x}^{t})_{t \in [k]}$ is a pure Nash equilibrium for $\game$, where $\hat{x}^t$ is arbitrary pure Nash equilibrium for $\game_t$.
\end{theorem}
\begin{proof}
For each $t \in [k]$, we denote by $\mathcal{M}_t = (L_t, d_t)$ the $t$-th finite metric space.
Note that each player $i = (i_1, \dots, i_k) \in V$  plays a $k$-tuple $(x_{i_{1}}^{1}, \dots, x_{i_k}^{k}) \in L_1 \times \cdots \times L_k$ as her strategy.
For a strategy profile $x = (x_{i})_{i \in V}$ of $\game$, the cost $c_i$ for a player $i = (i_1, \dots, i_k) \in V$ is 
\begin{align*}
	c_{i}(x) = \alpha d(x_{i}, \beta_{i}) + ( 1 - \alpha ) \sum_{j \in N(i)} d(x_{i}, x_{j}).
\end{align*}

For each $t \in [k]$ and each player $i_{t} \in V_{t}$, we denote by $c_{i_{t}}^{t}$ the cost function for $i_t$ on $\game_{t}$. Note that for a strategy profile $x^{t} = (x_{i_{t}}^{t})_{i_{t} \in V_{t}}$, the cost for a player $i_{t} \in V_{t}$ is 
\begin{align*}
	c_{i_t}^{t}(x^{t}) = \alpha d_{t}(x_{i_{t}}^{t}, \beta_{i_{t}}^{t}) + ( 1 - \alpha) \sum_{j_{t} \in N^{t}(i_{t})} d_{t}(x_{i_{t}}^{t}, x_{j_{t}}^{t}),
\end{align*}
where $N^{t}(i_t)$ is the set of neighbors of the player $i_t$ on the graph $G_{t}$. 

As mentioned above, given a strategy profile $x^{t} = (x_{i_t}^{t})_{i_t \in V_{t}}$ of the game $\game_{t}$ for each  $t \in [k]$, we interpret the tuple $x = (x^{t})_{t \in [k]}$ as the strategy profile of $\game$ such that each player $i = (i_1, \dots, i_k) \in V$ plays the strategy $x_{i} = (x_{i_1}^{1}, \dots, x_{i_k}^{k})$.
In this case, the cost for a player $i = (i_1, \dots, i_k)$ holds that 

\begin{align*}
	c_{i}&(x) 
		= \alpha d(x_{i}, \beta_{i}) + ( 1 - \alpha ) \sum_{j \in N(i)} d(x_{i}, x_{j})\\
		&= \sum_{t \in [k]} \alpha d_{t}(x_{i_{t}}^{t}, \beta_{i_{t}}^{t}) + (1-\alpha) \sum_{s \in [k]} \sum_{j \in N(i \mid s)} \sum_{t \in [k]} d_{t}(x_{i_{t}}^{t}, x_{j_{t}}^{t})\\
		&= \sum_{t \in [k]} \alpha d_{t}(x_{i_{t}}^{t}, \beta_{i_{t}}^{t})\\&~~~~~~~ + (1-\alpha) \sum_{s \in [k]} \sum_{j \in N(i \mid s)} \left( d_{s}(x_{i_{s}}^{s}, x_{j_{s}}^{s}) + \sum_{t \neq s} d_{t}(x_{i_{t}}^{t}, x_{j_{t}}^{t}) \right)\\
		&= \sum_{t \in [k]} \alpha d_{t}(x_{i_{t}}^{t}, \beta_{i_{t}}^{t}) + (1 -\alpha) \sum_{s \in [k]} \sum_{j \in N(i \mid s)} d_{s}(x_{i_{s}}^{s}, x_{j_{s}}^{s})\\
		&= \sum_{t \in [k]} \left( \alpha d_{t}(x_{i_{t}}^{t}, \beta_{i_{t}}^{t}) + (1 - \alpha) \sum_{j_{t} \in N^{t}(i_{t})} d_{t}(x_{i_{t}}^{t}, x_{j_{t}}^{t}) \right)\\
		&= \sum_{t \in [k]} c_{i_{t}}^{t}(x^t),
\end{align*}

\noindent
where $N(i \mid t) = \{ j \in N(i) ~;~ j_{t} \neq i_{t} \}$, which is the subset of the neighbors of $i$ on $G$ that are adjacent to $i$ by an edge on $G_{t}$.
The fourth equality follows from the construction of the strategy profile $x = (x^{t})_{t \in [k]}$.
To show the fifth equality, we use the fact that $\sum_{j \in N( i \mid t)}d_{t}(x_{i_{t}}^{t}, x_{j_{t}}^{t}) = \sum_{j \in N^{t}( i_{t})}d_{t}(x_{i_{t}}^{t}, x_{j_{t}}^{t})$.

Here, we prove that if a strategy profile $\hat{x}^{t}$ is a pure Nash equilibrium for $\game_{t}$ for each $t \in [k]$, then the strategy profile $\hat{x} = (\hat{x}^{t})_{t \in [k]}$ is a pure Nash equilibrium for $\game$.

For the sake of a contradiction, we assume that some player $i = (i_{1}, \dots, i_{k}) \in V$ can improve her cost by moving her strategy to $y_{i} = (y_{i_1}^{1}, \dots, y_{i_k}^{k})$, i.e., it satisfies that $c_{i}(\hat{x}_{i}, \hat{x}_{-i}) > c_{i}(y_{i}, \hat{x}_{-i})$.
Then we have
\begin{align*}
	0 &< c_{i}(\hat{x}_{i}, \hat{x}_{-i_{t}}) - c_{i}(y_{i}, \hat{x}_{-i_{t}})\\
		&= \sum_{t \in [k]} c_{i_{t}}^{t}(\hat{x}_{i_{t}}^{t}, \hat{x}_{ - i_{t}}^{t})\\&~~~~~~~  - \bigg( \sum_{t \in [k]} \alpha d_{t}(y_{i_{t}}^{t}, \beta_{i_{t}}^{t}) + (1-\alpha) \sum_{s \in [k]} \sum_{j \in N(i \mid s)} \sum_{t \in [k]} d_{t}(y_{i_{t}}^{t}, \hat{x}_{j_{t}}^{t}) \bigg)\\
		&= \sum_{t \in [k]} c_{i_{t}}^{t}(\hat{x}_{i_{t}}^{t}, \hat{x}_{ - i_{t}}^{t})\\&~~~~~~~ - \bigg( \sum_{t \in [k]} \alpha d_{t}(y_{i_{t}}^{t}, \beta_{i_{t}}^{t})+ (1-\alpha) \sum_{s \in [k]} \sum_{j \in N^{s}(i_{s})}d_{s}(y_{i_s}^{s}, \hat{x}_{j_{s}}^{s}) \bigg) - \delta\\
		&= \sum_{t \in [k]} c_{i_{t}}^{t}(\hat{x}_{i_{t}}^{t}, \hat{x}_{ - i_{t}}^{t}) - \sum_{t \in [k]} c_{i_{t}}^{t}(y_{i_t}^{t}, \hat{x}_{-i_t}^{t}) - \delta\\
		&= \sum_{t \in [k]} \left( c_{i_{t}}^{t}(\hat{x}_{i_{t}}^{t}, \hat{x}_{-i_{t}}^{t}) - c_{i_{t}}^{t}(y_{i_{t}}^{t}, \hat{x}_{-i_{t}}^{t}) \right) - \delta,
\end{align*}
where $\delta = (1-\alpha) \sum_{s \in [k]} \sum_{j \in N(i \mid s)} \sum_{t \neq s} d_{t}(y_{i_{t}}^{t}, \hat{x}_{j_{t}}^{t})$. Note that $\delta$ is non-negative.

Recall that for every $t \in [k]$, the strategy profile $\hat{x}^{t}$ is a pure Nash equilibrium for $\game_t$. Thus, it holds that $c_{i_{t}}^{t}(\hat{x}_{i_t}^{t}, \hat{x}_{-i_{t}}^{t}) \le c_{i_{t}}(y_{i_{t}}^{t}, \hat{x}_{-i_{t}}^{t})$. Therefore, we have
\begin{align*}
	0 \le \delta < \sum_{t \in [k]} \left( c_{i_{t}}^{t}(\hat{x}_{i_{t}}^{t}, \hat{x}_{-i_{t}}^{t}) - c_{i_{t}}^{t}(y_{i_{t}}^{t}, \hat{x}_{-i_{t}}^{t}) \right) \le 0,
\end{align*}
 which is a contradiction.
\end{proof}

\subsection{Polynomial-time Solvability of Discrete Preference Games} \label{SecPolynomial-timeSolvabilityDPG}
In this section, we prove Theorem \ref{TheoremPolynomial-timeComputabilityGridGraph}.
From the condition \ref{ConditionA}, the strategy space $\mathcal{M}$ on $\game$ is a $1$-product metric space of $k$ finite metric spaces $\mathcal{M}_{1}, \dots, \mathcal{M}_{k}$. 
From the condition \ref{ConditionB}, we can select a strategy $\beta_{i_{t}}^{t} \in L_{t}$ for each $t \in [k]$ and each $i_{t} \in [M_{t}]$ so that the set $\{ \beta_{i_{t}}^{t} \in L_{t} ~;~ t \in [k], i_{t} \in [M_{t}] \}$ such that for each player $i = (i_1, \dots, i_{k}) \in V$, the preferred strategy $\beta_{i}$ is equal to $(\beta_{i_{1}}^{1}, \dots, \beta_{i_{k}}^{k})$.

Note that a $k$-dimensional grid graph $G = ([M_1] \times \cdots \times [M_k], E)$ is a cartesian product of graphs $G_{1} = ( [M_{1}], E_{1}), \dots,$ $G_{k} = ([M_{k}], E_{k})$, where there is an edge $\{ i, j\} \in E_{t}$ if $| i - j | = 1$ for each $t \in [k]$. We use this fact to prove this theorem.

We decompose $\game$ into $k$ discrete preference games $\game_{1}, \dots, \game_{k}$ such that $\game$ is to be a Cartesian game constructed from these subgames.
For each $t \in [k]$, we defne the $t$-th subgame as $\game_{t} = (G_{t}, \mathcal{M}_{t}, (\beta_{i_t}^{t})_{i_t \in [N_t]}, \alpha)$.
It is easy to see that $\game$ is a product game constructed from $k$ discrete preference games $\game_{1}, \dots, \game_{k}$.

Note that for every $t \in [k]$, the players' network on $\game_{t}$ has $O(\log n)$-treewidth.
Therefore, we can find a pure Nash equilibrium $\hat{x}^{t}$ for $\game_{t}$ in polynomial time from the result by Daskalakis and Papadimitriou \cite{DP06}.
By Theorem \ref{TheoremCartesianProductDPG}, the strategy profile $\hat{x} = (\hat{x}^{t})_{t \in [k]}$ constructed from pure Nash equilibria $\hat{x}^{1}, \dots, \hat{x}^{k}$ is a pure Nash equilibrium for $\game$. Therefore, we can compute a pure Nash equilibrium for $\game$ in polynomial time. 

\subsection{Properties of Discrete Preference Games on Product Metric Spaces} \label{SecGameProductMetricSpace}
In the rest of this section, we focus on a discrete preference game on a product metric space. Recall that Lolakapuri et al. \cite{LBNPD19} considered a $1$-product metric space of some path metric spaces and have proven that the problem of finding a pure Nash equilibrium for a discrete preference game on such a metric space is polynomial-time computable.
Their algorithm, called {\it Product Metric Algo} produced in \cite{LBNPD19}, gives us an approach to computing pure Nash equilibria for games: It may be easier to compute it when we can decompose the strategy space into an $\ell$-product metric space for some $\ell \in \positiveint \cup \{ \infty \}$.

This section discusses the conditions under which such an approach, a decomposition approach, would work well. We prove that the decomposition approach always works for a discrete preference game on a $1$-product metric space of arbitrary finite metric spaces.

Before discussing, we describe the more general model of discrete preference games, introduced by Lolakapuri et al.\ \cite{LBNPD19}.
In their model, a game has edge weights and a penalty for each strategy instead of a parameter.
A {\it discrete preference game with penalties} $\game = (G, \mathcal{M}, ( p_{i}(s) ) )$ is defined by: (i) an edge-weighted graph $G = (V, E, (w_e)_{e \in E})$; (ii) each player $i \in V$ has a penalty $p_{i}(s) \in \nonnegreal$ for each strategy $s \in L$, where $L$ is a finite set of strategies; (iii) given a strategy profile $x = (x_i)_{ i \in V}$, the cost for player $i \in V$ is:
\begin{align}
	c_i(x) = \sum_{s \in L}	p_i(s) d(x_i, s) + \sum_{ j \in N(i)} w_{ij} d(x_i, x_j).
\end{align}

Let $( G = (V, E, (w_e)_{e \in E}), \mathcal{M} = (L, d), (p_{i}(s))_{i \in V, s \in L})$ be a discrete preference game.
For this game, we define the function $\Phi \colon L^V \to \nonnegreal$ as follows: 
\begin{align} \label{Eq:ExactPotential-DPS}
	\Phi(x) = \sum_{ i \in V} \sum_{ s \in L} p_i(s) d(s, x_i) + \sum_{ \{ i, j \} \in E} w_{ij} d(x_i, x_j).
\end{align}
We show that $\Phi$ is an exact potential function for a discrete preference game with penalties. 

\begin{lemma} \label{LemmaDPG-is-EPG}
Let $\game = \left( G, \mathcal{M}, ( p_{i}(s) ) \right)$ be a discrete preference game with penalties, where $G = (V, E, (w_e)_{e \in E})$,  $\mathcal{M} = (L, d)$. The game $\game$ is an exact potential game.
\end{lemma}
\begin{proof}
To see why the function $\Phi$ defined as Eq.\ (\ref{Eq:ExactPotential-DPS}) is an exact potential function for $\game$, for each player $i \in V$, all two strategies $x_i$ and $y_i$, and all strategies $x_{-i}$ of all players expect $i$, it holds that 
\begin{align*}
		\Phi(x_i, x_{-i}) &- \Phi(y_i, x_{-i}) \\
		&= \sum_{ s \in L} p_i(s) d(s, x_i) + \sum_{ j \in N(i)} w_{ij} d(x_i, x_j)\\ &~~~~~~~  - \left(\sum_{ s \in L} p_i(s) d(s, y_i) + \sum_{ j \in N(i)} w_{ij} d(y_i, x_j) \right)\\
		&= c_i(x_i, x_{-i}) - c_i(y_i, x_{-i}).
\end{align*}
Thus, a discrete preference game with penalties is an exact potential game.	
\end{proof}

Now, we consider a discrete preference game on an $\ell$-product metric space.
Let $\game = (G = (V, E, (w_e)_{e \in E}), \mathcal{M} = (L, d), ( p_{i}(s) )_{i \in V, s \in L})$ be a discrete preference game with penalties, and let $\ell \in \positiveint \cup \{ \infty \}$.
Suppose that $\mathcal{M}$ is an $\ell$-product metric space formed by $k$ finite metric spaces $\mathcal{M}_1 = (L_1, d_1), \dots, \mathcal{M}_k = (L_k, d_k)$.
For a strategy profile $x = (x_{i})_{i \in V}$ on $\game$, we interpret it as that each player $i \in V$ plays $k$-tuple $x_{i} = (x_{i}^{1}, \dots, x_{i}^{k}) \in L_{1} \times \cdots \times L_{k}$. We denote by $x^{t} = (x_{i}^{t})_{i \in V}$ the list on  $L_{t}$ for each strategy profile $x$ on $\game$.
For a strategy profile $x = (x^t)_{t \in [k]}$ on $\game$, we interpret the strategy $x_{i}$ of player $i \in V$ as the $k$-tuple of strategies $x_{i} = (x_{i}^1, \dots, x_{i}^k)$, where $x^t$ is a strategy profile on $L_t$ for each $t \in [k]$.

We decompose $\game$ into $k$ discrete preference games on the partial metric spaces. In the following, we refer to such games as subgames.
For each $t \in [k]$, the $t$-th subgame $\game_{t}$ of $\game$ is defined as $\game_{t} := ( G, \mathcal{M}_t, (q_{i}^{t}(s^{t}))_{i \in V, s^t \in L_t})$, where $q_{i}^{t}(s^{t}) = \sum_{u \in L \colon u^t = s^t}p_i(u)$.
Then the cost $c_{i}^{t}(x^t)$ of a player $i \in V$ on the $t$-th subgame is

\begin{align}
	c_{i}^{t}(x^t) = \sum_{ s^t \in L_t} q_{i}^{t}(s^{t}) d_{t}(x_{i}^{t}, s^{t}) + \sum_{ j \in N(i)} w_{ij} d_{t}(x_{i}^{t}, x_{j}^{t}). 
\end{align}
We denote by $\Phi^{(t)}(x^{t})$ the exact potential function for the $t$-th subgame, i.e.,
\begin{align*}
	\Phi^{(t)}(x^{t}) = \sum_{i \in V}\sum_{ s^t \in L_t} q_{i}^{t}(s^{t}) d_{t}(x_{i}^{t}, s^{t}) + \sum_{ \{i, j\} \in E} w_{ij} d_{t}(x_{i}^{t}, x_{j}^{t}).
\end{align*}

Furthermore, we define a function $\Psi(x)$ as
\begin{align}
	\Psi(x) = \sum_{t \in [k]} \Phi^{(t)}(x^{t}).
	\label{Eq:GeneralizedOrdinalPotentialFunction}
\end{align}

\begin{theorem} \label{TheoremSufficientCondition}
	If the function $\Psi$ defined in Eq.\ (\ref{Eq:GeneralizedOrdinalPotentialFunction}) is a generalized ordinal potential function for $\game$, then a strategy profile $\hat{x} = (\hat{x}^{t})_{t \in [k]}$ is a pure Nash equilibrium for $\game$, where $\hat{x}^t$ is an arbitrary pure Nash equilibrium for the $t$-th subgame.
\end{theorem}
\begin{proof}
	For the sake of a contradiction, we assume that $\hat{x}$ is not a pure Nash equilibrium for $\game$, and thus, there is a player $i$ that can improve her cost by moving to another strategy $y_i$ from $\hat{x}_i$.
	Then it holds that $c_i(\hat{x}_i, \hat{x}_{-i}) > c_{i}(y_i, \hat{x}_{-i})$.
	Since $\Psi$ is a generalized ordinal potential function for $\game$, it satisfies that
	\begin{align*}
		0 < \Psi(\hat{x}_i, \hat{x}_{-i}) - \Psi(y_i, \hat{x}_{-i})
			= \sum_{t \in [k]} \big( c_{i}^{t}(\hat{x}_i^t, \hat{x}_{ - i}^t) - c_{i}^{t}(y_i^t, \hat{x}_{ - i}^t) \big).
	\end{align*}
	This implies that there is at least one $t \in [k]$ such that  $c_{i}^{t}(\hat{x}_i^t, \hat{x}_{ - i}^t) > c_{i}^{t}(y_i^t, \hat{x}_{ - i}^t)$. Note that for each $t \in [k]$, $\hat{x}^t$ is a pure Nash equilibrium for the $t$-th subgame, and hence, we have  $c_{i}^{t}(\hat{x}_i^t, \hat{x}_{ - i}^t) \le c_{i}^{t}(y_i^t, \hat{x}_{ - i}^t)$. This is a contradiction.
\end{proof}

\begin{corollary}
	There is a polynomial-time algorithm to find a pure Nash equilibrium for $\game$ when the following two conditions hold: (i) for each $t$-th subgame, we have a polynomial-time algorithm to find a pure Nash equilibrium; and (ii) the function $\Psi$ is a generalized ordinal potential function for $\game$.
\end{corollary}

Unfortunately, the function $\Psi$ is not always a generalized ordinal potential function for $\game$.
Consider a discrete preference game on a discrete metric space with $2^k$ points. Note that such a metric can be written straightforwardly as an $\infty$-product metric of $k$ discrete metrics. Example \ref{Example} shows that the metric decomposition approach is not easily applicable in such a game.

\begin{example} \label{Example}
	For simplicity, we consider a discrete preference game with a parameter.
	Let $G = (V, E)$  be an unweighted graph, and $\mathcal{M}$ be a discrete metric space on $2^k$ strategies. For each player $v \in V$, we denote by $\beta_v \in [2^k]$ the preferred strategy of player $v$. Furthermore, we are given a parameter $1/2 < \alpha < 1$.
	
	We decompose $\mathcal{M}$ into $k$ metric spaces $( \{0, 1\}, \delta)$. Each point $x \in [2^k]$ is interpreted as the binary string, and hence, the point on the $t$-th metric is the the $t$-th bit for $x$.
	Here, the function $\delta$ is also the discrete metric, i.e., $\delta(x, y) = 1$ if $x \neq y$, otherwise $\delta(x, y) = 0$.
	It is easy to see that for each pair of points $x, y \in L$, it satisfies that $d(x, y) = \max_{t \in [k]} \delta(x^t, y^t)$, where $x^t$ is the $t$-th bit of $x$.
	The cost for player $i \in V$ on the $t$-th subgame is $c_{i}^{t}(x^t) = \alpha \delta(\beta_{i}^{t}, x_{i}^{t}) + (1-\alpha) \sum_{ j \in N(i)} \delta( x_{i}^{t}, x_{j}^{t} )$.
	
	Now, we show that this game does not satisfy the condition of Theorem \ref{TheoremSufficientCondition}.
	We fix any player $i \in V$. Then, we take strategies $x_i$ and $y_i$ for $i$ and strategies $x_{ - i }$ for all others except $i$ such that it satisfies the following conditions:
	\begin{itemize}
		\item $x_i$ and $y_i$ are different at only the $t$-th bit for some $t \in [k]$;
		\item $x_{i}^{t} = \beta_{i}^{t}$ and $x_{i} \neq \beta_{i} \neq y_{i}$;
		\item 	$D_{i}(x_{i}, x_{-i}) > D(y_i, x_{-i})$; and
		\item $D_{i}^{t}(x_{i}^{t}, x_{ -i }^{t}) - D_{i}^{t}(y_{i}^{t}, x_{ -i }^{t}) \le 1$,
	\end{itemize}
	where $D_{i}(x_{i}, x_{-i})$ denotes that the number of $i$'s neighbors that play a different strategy from $x_{i}$, and also we denote $D_{i}^{t}(x_{i}^{t}, x_{-i}^{t})$ the number of $i$'s neighbors whose $t$-th strategy is not $x_{i}^{t}$.
	
	In this setting, the player $i$ can decrease her cost by moving $x_{i}$ to $y_i$.
	On the other hand, for the cost $c_{i}^{t}$ for $i$ on the $t$-th subgame, it follows that 
	\begin{align*}
		c_{i}^{t}(x_{i}^{t}, x_{ -i }^{t}) &- c_{i}^{t}(y_{i}^{t}, x_{ -i }^{t})\\ 
			&= - \alpha + ( 1 - \alpha) \left( D_{i}^{t}(x_{i}^{t}, x_{-i}^{t}) - D_{i}^{t}(y_{i}^{t}, x_{-i}^{t}) \right)\\
			&\le - \alpha + ( 1 -\alpha )
			= 1 - 2 \alpha
			< 0 .
	\end{align*}
	The first equality holds from the second assumption, and note that $y_{i}^{t} \neq x_{i}^{t}$ in this setting.
	The second inequality follows from the fourth assumption. The final inequality follows from $1/2 < \alpha < 1$.
	
	The above observation implies that $i$ can not improve her cost in the $t$-th subgame.
	Notice that $i$ moves only one bit from the first assumption, the function defined in Eq.\ (\ref{Eq:GeneralizedOrdinalPotentialFunction}) is not a generalize ordinal potential function.
	\qed
\end{example}

The next theorem states that for a discrete preference game $\game$ on a $1$-product metric space, the function $\Psi$ defined in Eq.\ (\ref{Eq:GeneralizedOrdinalPotentialFunction}) is always a generalized ordinal potential function for $\game$; more precisely, $\Psi$ is an exact potential function for $\game$. To prove this theorem, it suffices to show that $\Psi$ equals $\Phi$ defined in Eq.\  (\ref{Eq:ExactPotential-DPS}).

\begin{theorem} \label{Theoremone-product-ExactPotentialGame}
	If a metric space $\mathcal{M}$ is a $1$-product metric space, then the function $\Psi$ defined in Eq.\ (\ref{Eq:GeneralizedOrdinalPotentialFunction}) is an exact potential function for $\game$.
\end{theorem}
\begin{proof}
We suppose that $\mathcal{M} = (L, d)$ is a $1$-product metric space of $k$ metric spaces, i.e., $d(x, y) = \sum_{t \in [k]} d_{t}(x^{t}, y^{t})$ for all $x, y \in L$.
It suffices to show that $\Phi$ defined in Eq.\ (\ref{Eq:ExactPotential-DPS}) equals to $\Psi$ defined in Eq.\ (\ref{Eq:GeneralizedOrdinalPotentialFunction}).
For any strategy profile $x$, we have

\begin{align*}
	\Phi(x) &= \sum_{ i \in V} \sum_{ s \in L} p_{i}(s) d(s, x_{i}) + \sum_{ \{ i, j \} \in E} w_{ij} d(x_i, x_j)\\
					&= \sum_{ i \in V} \sum_{ s \in L} p_{i}(s) \sum_{ t \in [k]} d_t(s^t, x_{i}^t) + \sum_{ \{ i, j \} \in E} w_{ij} \sum_{ t \in [k]} d_t(s^t, x_i^t)\\
					&= \sum_{ t \in [k]} \sum_{ i \in V} \sum_{ s^t \in L_t} \sum_{ u \in L \colon u^t = s^t}p_{i}(u) d_t(s^t, x_{i}^t) + \sum_{ t \in [k]} \sum_{ \{ i, j \} \in E} w_{ij} d_t(s^t, x_i^t)\\
					&= \sum_{ t \in [k]} \Phi^{(t)}
					= \Psi(x).
\end{align*}

\noindent
In the second equality, we use the fact that $\mathcal{M}$ is a $1$-product metric space.
\end{proof}

Immediately, we obtain the following corollary, which is a generalization of the result by Lolakapuri et al.\ \cite{LBNPD19}.

\begin{corollary} \label{CorollaryGenralization-LBNPD}
	There is a polynomial-time algorithm to find a pure Nash equilibrium for a discrete preference game if the following two conditions hold: (i) the metric space is a $1$-product metric space; and (ii) we have a polynomial-time algorithm to find a pure Nash equilibrium for every subgame.
\end{corollary}

\section{Relationship between Network Coordination Games and Discrete Preference Gamess} \label{SecRelationshipTheseGames}
This section presents the relationship between network coordination games and discrete preference games.
First, we show that every discrete preference game is polynomial-time reducible to a network coordination game.
Second, we provide a class of network coordination games that are polynomial-time reducible to discrete preference games.

\subsection{Reduction from Discrete Preference Games to Network Coordination Games} \label{SecReductionfromDPGtoNCG}
This section shows that a discrete preference game is reducible to a network coordination game in polynomial time.
\begin{lemma} \label{LemmaPolynomial-timeComputability}
	Let $\game$ be a discrete preference game on a graph $G$.
	If we have a polynomial-time algorithm to compute pure Nash equilibria for network coordination games on the graph $G$, then it is also polynomial-time computable to find a pure Nash equilibrium for $\game$.
\end{lemma}
\begin{proof}
	To prove this, it is sufficient to construct a polynomial-time reduction from a discrete preference game to a network coordination game that preserves the structure of the players' network.

	For each player $i \in V$, we denote as $\Delta_i := |N(i)|$.
	For each edge $e = \{ i, j \} \in E$, we define the cost function $C_{i, j}$ as follows: for each element $(x_i, x_j) \in L \times L$,
	\begin{align*}
		C_{i, j}(x_i, x_j) = \sum_{k \in \{i, j\}} \Delta_{k}^{-1} \sum_{s \in L} {p_{k} d(s, x_k)} + w_{e} d(x_i, x_j)
	\end{align*}
	which means the cost for $i$ and $j$ when $i$ plays $x_i$ and $j$ plays $x_j$.
	
	The exact potential function $\Phi'$ for a network coordination game is
	\begin{align}
		\Phi'(x) = \sum_{ e = \{i, j\} \in E} C_{i, j}(x_i, x_j)
		\label{Eq:ExactPotentialFunction-of-NCG}
	\end{align}
	for each strategy profile $x = (x_i)_{ i \in V}$ \cite{CD11}.
	
	To see that every pure Nash equilibrium for our network coordination game is also a pure Nash equilibrium for the given discrete preference game, we show that $\Phi'$ is also an exact potential function for a discrete preference game (see Theorem 2.2 in Chapter 2 of \cite{LHS16}).
	
	\begin{align*}
		\Phi'(x) 	&= \sum_{ e = \{i, j\} \in E} C_{i, j}(x_i, x_j)\\
						&= \sum_{ e = \{i, j\} \in E} \sum_{k \in e } \Delta_{k}^{-1} \sum_{s \in L}{ p_k(x) d(s, x_k) } + \sum_{e = \{i, j\} \in E} w_{i, j} d(x_i, x_j)\\ 
						&= \sum_{ i \in V} \left( \sum_{j \in N(i)} \sum_{s \in L}{ \Delta_{i}^{-1} p_i(s) d(s, x_i)} \right) + \sum_{ e = \{ i, j \} \in E} w_{i, j}d(x_i, x_j)\\
						&= \sum_{ i \in V} \sum_{s \in L}{p_i(s) d(s, x_i)} + \sum_{ e = \{ i, j \} \in E} w_{i, j}d(x_i, x_j)\\
						&= \Phi(x).
	\end{align*}
	
	\noindent
	This is an exact potential function for a discrete preference game (see Eq.\ (\ref{Eq:ExactPotential-DPS})).
	Hence, we complete constructing a polynomial-time reduction from a discrete preference game to a network coordination game. 
	Note that our reduction does not change the structure of the graph $G$.
\end{proof}

Recall that Daskalakis and Papadimitriou \cite{DP06} have proven that we can find a pure Nash equilibrium for a graphical game whose players' network has $O(\log n)$-treewidth in polynomial time.
Apt et al.\ \cite{AKRSS17} have shown the polynomial-time computability of a pure Nash equilibrium for a network coordination game whose players' network contains at most one cycle.
Therefore, we immediately obtain the following corollary by using these previous results together with Lemma \ref{LemmaPolynomial-timeComputability}.

\begin{corollary}
	There is a polynomial-time algorithm to compute a pure Nash equilibrium for a discrete preference game if the given players' network $G = (V, E)$ satisfies at least one of the following properties: (i) $G$ has $O(\log |V|)$-treewidth; and (ii) $G$ contains at most one cycle.
\end{corollary}

\subsection{Reduction from Network Coordination Games to Discrete Preference Games}
In the previous section, we show that a discrete preference game is a special case of network coordination games.
This section provides a class of network coordination games that are polynomial-time reducible to discrete preference games.
Note that it is known that equilibrium computation for our class of network coordination games is easy by using a submodular function minimizing algorithm, such as \cite{LSW15, Orl09}.
However, we solve equilibrium computation faster by reducing a discrete preference game (see Remark \ref{RemarkSubmodularMinimization} for details).

We consider the complexity of a two-strategic network coordination game such that for each pair of players $i, j$, the cost $C_{i, j}$ between $i$ and $j$ is symmetric, i.e., $C_{i, j}(0, 1) = C_{i, j}(1, 0)$ and a submodular function, i.e., 
\begin{align} \label{Eq:Submodular}
	 C_{i, j}(1, 0) + C_{i, j}(0, 1) \ge C_{i, j}(1, 1) + C_{i, j}(0, 0), 
\end{align}
where we denote by $\{0, 1\}$ the set of strategies.

In this setting, we show that we can find a pure Nash equilibrium in $O(n^2 \Delta)$ time by reducing it to a discrete preference game on a path metric space, where $n$ is the number of players, and $\Delta$ is the maximum degree of a given graph.

\begin{theorem} \label{TheoremEfficientSymmetricNetworkCoordinationGame}
Suppose that a two-strategic network coordination game $\game = ( G = (V, E), ( \{0, 1\})_{v \in V}, (C_e)_{e \in E} )$ satisfies that for each edge $\{i, j\} \in E$, a cost function $C_{i, j}$ is a symmetric submodular function. In this setting, we can find a pure Nash equilibrium for $\game$ in  $O(n^2 \Delta)$ time, where $n$ is the number of players, and $\Delta$ is the maximum degree of $G$.
\end{theorem}
\begin{proof}
Let $\game = ( G = (V, E), ( \{0, 1\})_{v \in V}, (C_e)_{e \in E} )$ be a network coordination game.
We now reduce this game to a discrete preference game on the path metric space $\mathcal{M} = ( \{0, 1\}, d)$, where $d(x, y) = 1/2$ if $x \neq y$, otherwise $d(x, y) = 0$.
Furthermore, our reduction preserves the construction of the players' network; hence, the resulting discrete preference game is on the graph $G = (V, E)$.

For each edge $\{ i, j \} \in E$, the weight is $w_{i, j} = 2 C_{i, j}(1, 0) - C_{i, j}(0, 0) - C_{i, j}(1, 1)$.
For each player $i \in V$, the penalty is $p_{i}(s) = \sum_{j \in N(i)} {C_{i, j}(1-s, 1-s)}$ for each $s \in \{0, 1\}$.
Note that every weight $w_{i, j}$ on an edge $\{i, j\} \in E$ is non-negative from our restrictions.

We denote by  $\game' = ( G' = ( V, E, (w_{e})_{e \in E}), \mathcal{M}, (p_{v}(0), p_{v}(1))_{v \in V})$ the resulting discrete preference game with penalties.
From Lemma \ref{LemmaDPG-is-EPG}, the exact potential function $\Phi$ for $\game'$ is
\begin{align}
	 \Phi(x) = \sum_{i \in V} \sum_{s \in \{0, 1\}} p_{i}(s) d(s, x_{i}) + \sum_{ \{ i, j \} \in E} w_{i, j} d(x_i, x_i).
	 \label{Eq:Reduction-NCG-to-DPG}
\end{align} 

We show, in Lemma \ref{LemmaReduction-NCG-to-DPG}, that $\Psi$ equals the exact potential function defined in Eq.\ (\ref{Eq:ExactPotentialFunction-of-NCG}) for the given network coordination game $\game$.
Proving this, we complete the reduction from $\game$ to a discrete preference game on a path metric space.

\begin{lemma} \label{LemmaReduction-NCG-to-DPG}
	The above function $\Phi$ is an exact potential function for $\game$.
\end{lemma}
\begin{proof}
It suffices to show that $\Phi$ defined in Eq.\ (\ref{Eq:Reduction-NCG-to-DPG}) equals the function defined in  Eq.\ (\ref{Eq:ExactPotentialFunction-of-NCG}).
By definition, it follows that 
\begin{align*}
	\Phi &(x) 
		= \sum_{i \in V} \sum_{k = 0, 1} p_{i}(k) d(k, x_{i}) + \sum_{ \{ i, j \} \in E} w_{i, j} d(x_i, x_j)\\
		&= \sum_{i \in V} \sum_{j \in N(i)} \left( C_{i, j}(1, 1) d(0, x_{i}) + C_{i, j}(0, 0) d(1, x_{i}) \right)\\ 
		&	+ \frac{1}{2} \sum_{i \in V}\sum_{ j \in N(i) } \big( 2 C_{i, j}(0, 1) - C_{i, j}(1, 1) - C_{i, j}(0, 0) \big) d(x_i, x_j)\\
		&= \frac{1}{2} \sum_{i \in V} \sum_{j \in N(i)} \big(  C_{i, j}(1, 1) d(0, x_{i}) + C_{i, j}(0, 0) d(1, x_{i})\\
		& +  C_{i, j}(1, 1) d(0, x_{j}) + C_{i, j}(0, 0) d(1, x_{j}) \big)\\
		& + \frac{1}{2} \sum_{i \in V}\sum_{ j \in N(i) } \big( 2 C_{i, j}(0, 1) - C_{i, j}(1, 1) - C_{i, j}(0, 0) \big) d(x_i, x_j)\\
		&= \frac{1}{2} \sum_{ i \in V} \sum_{j \in N(i)} \big(  C_{i, j}(1, 1) d(0, x_{i}) + C_{i, j}(0, 0) d(1, x_{i}) \\
		& +  C_{i, j}(1, 1) d(0, x_{j}) + C_{i, j}(0, 0) d(1, x_{j})\\
		& + 2 C_{i, j}(0, 1) d(x_i, x_j) - C_{i, j}(1, 1) d(x_i, x_j)\\& - C_{i, j}(0, 0) d(x_i, x_j) \big)\\
		&= \frac{1}{2} \sum_{ i \in V} \sum_{j \in N(i)} \bigg( C_{i, j}(0, 0) \big( d(1, x_i) + d(1, x_j) - d(x_i, x_j) \big) \\
		&~~~~~~~~~ + C_{i, j}(1, 1) \big( d(0, x_i) + d(0, x_j) - d(x_i, x_j) \big)\\&~~~~~~~~~ + 2 C_{i, j}(0, 1)  d(x_i, x_j) \bigg)\\
		&= \frac{1}{2} \sum_{i \in V} \sum_{j \in N(i)} C_{i, j}(x_i, x_j).
\end{align*}
Note that, in the third equality, for each player $i \in V$, we add the additional value $C_{i, j}(1, 1) d(0, x_{j}) + C_{i, j}(0, 0) d(1, x_{j})$ for every neighbor $j \in N(i)$. That value also appears in the terms of $i$ in the third equation. 
Dividing the new summation
\begin{dmath*}
\sum_{i \in V} \sum_{j \in N(i)} \bigg( C_{i, j}(1, 1) d(0, x_{i}) + C_{i, j}(0, 0) d(1, x_{i}) +  C_{i, j}(1, 1) d(0, x_{j}) + C_{i, j}(0, 0) d(1, x_{j})  \bigg)
\end{dmath*}
into half, it is equivalent to the the first summation of the third equation. Therefore, the fourth equality holds.
The final equality follows from the following fact: for each $s \in \{0, 1\}$ and each pair of $i, j \in V$,
\begin{align*}
	d(s, x_i) + d(s, x_j) - d(x_i, x_j) = \begin{cases}
 		1 &\mbox{ if } x_i = x_j = 1 - s\\
 		0 & \mbox{ otherwise.}
 \end{cases}
\end{align*}

Furthermore, we have 
\begin{align*}
	\frac{1}{2} \sum_{i \in V} \sum_{j \in N(i)} C_{i, j}(x_i, x_j) = \sum_{ \{ i, j \} \in E} C_{i, j}(x_i, x_j)	
\end{align*}
which is the exact potential function for the network coordination game $\game$, defined in Eq (\ref{Eq:ExactPotentialFunction-of-NCG}).
Thus, we complete the proof of Lemma \ref{LemmaReduction-NCG-to-DPG}.
\end{proof}

Since any two-element finite metric space is a tree metric pace, we can apply {\it Tree Metric Algo},  proposed by Lolakapuri et al.\ \cite{LBNPD19}, to find a pure Nash equilibrium for $\game'$. {Here, \it Tree Metric Algo} is an algorithm for computing a pure Nash equilibrium for a discrete preference game on a tree metric space.

In a two-strategic setting, each player moves her strategy at most once during  {\it Tree Metric Algo}.
Recall that Lolakapuri et al.\ \cite{LBNPD19} showed the following theorem:
\begin{theorem}[Lolakapuri et al.\ \cite{LBNPD19}]
	For a discrete preference game that has $n$ players and whose tree metric space has $m$ points, the Tree Metric Algo outputs a pure Nash equilibrium on the given game, in $O(nm \cdot n \mathrm{EO})$-time. Here $\mathrm{EO}$ is the time to evaluate the cost function for a player.
\end{theorem}
\noindent
Since the metric space has only two points and the cost function for each player can be evaluated in $O(\Delta)$-time, we can compute a pure Nash equilibrium for a $\game'$ in $O(n^2 \Delta)$ time.

Since each pure Nash equilibrium for $\game'$ agrees with a pure Nash equilibrium for the network coordination game $\game$, we obtain a pure Nash equilibrium for $\game$ from the above argument. We complete the proof of Theorem \ref{TheoremEfficientSymmetricNetworkCoordinationGame}.
\end{proof}

\begin{remark} \label{RemarkSubmodularMinimization}
If every cost function on a network coordination game is a submodular function, then the exact potential function $\Phi'$ defined in Eq.\ (\ref{Eq:ExactPotentialFunction-of-NCG}) is also a submodular function. This implies that we can apply an algorithm for submodular function minimization, such as \cite{LSW15,Orl09}, to find a pure Nash equilibrium. In particular, we solve it in $O(n^3 \log^{2}(n) \cdot \mathrm{EO} + n^4 \log^{O(1)}(n))$ time \cite{LSW15}, where $n$ is the number of players, and $\mathrm{EO}$ is the time to evaluate $\Phi'$, which is bounded by the number of edges.

We obtain a pure Nash equilibrium that minimizes the corresponding potential function by using the submodular minimization algorithm.
Note that equilibrium computation allows any pure Nash equilibrium as a solution; that is, a solution that we obtain does not necessarily minimize the corresponding function.
In Theorem \ref{TheoremEfficientSymmetricNetworkCoordinationGame}, we exploit this fact, and thus, we can find a pure Nash equilibrium faster than applying an algorithm for submodular function minimization.
Our result implies that equilibrium computation is solved at least $O(n)$ factor faster.
\end{remark}

\section{Conclusion}
We have studied the complexity of computing a pure Nash equilibrium for a discrete preference game on a grid graph. As mentioned in Section \ref{SecIntroduction}, our motive behind this work is to resolve the main open question for network coordination games: Is it tractable to find a pure Nash equilibrium for a network coordination game on a graph with degree four? Under negative conjecture, we study the complexity of computing a pure Nash equilibrium for a discrete preference game, a subclass of network coordination games.

Unfortunately, it is still open whether finding a pure Nash equilibrium for a discrete preference game on a graph with degree four is tractable. We have shown the polynomial-time computability for a discrete preference game with a parameter on a $k$-dimensional grid graph when it satisfies the two conditions \ref{ConditionA} and \ref{ConditionB}. It is the first result for efficient computability for a discrete preference game with neither $O(\log n)$-treewidth nor a tree metric space. Note that our result holds under somewhat artificial conditions.  An interesting open question worth considering is whether it is also tractable if we remove the condition \ref{ConditionA} or \ref{ConditionB}.

Another interesting direction would be the complexity of computing pure Nash equilibria for discrete preference games on the discrete metric space with three or more elements. We provide, in Section \ref{SecDiscreteMetric}, an upper bound for the number of iterations of the best response dynamics for a discrete preference game with a parameter on a discrete metric space. The discrete preference metric space with three or more strategies is one of the simple environments among finite metric spaces, not a tree metric space.

Finally, we have discussed the complexity of computing a two-strategic network coordination game whose cost functions are symmetric submodular functions. In this case, the game is reducible to a discrete preference game on a path metric space, and we can find a pure Nash equilibrium faster than an algorithm for submodular function minimization.
An open question worth considering is whether we can also compute a pure Nash equilibrium faster than an algorithm for submodular function minimization when a cost function is asymmetric.

\section*{Acknowledgments}
This work was supported by JSPS KAKENHI Grant Numbers JP21J10845 and JP20H05795.

\printbibliography[heading = bibintoc]

\end{document}